\documentclass[a4paper,fleqn]{cas-dc}

\usepackage{amsthm}
\usepackage{amsmath}
\usepackage{amsfonts}
\usepackage{multirow}
\usepackage{mathtools}
\usepackage{amssymb}
\usepackage{graphicx}
\usepackage{textcomp}
\graphicspath{{\images}}
\usepackage{siunitx}
\usepackage{ragged2e}
\usepackage{todonotes}
\usepackage{lipsum}
\theoremstyle{plain}
\usepackage{xhfill}
\usepackage{booktabs}         
\usepackage{courier}
\usepackage{algpseudocode}
\usepackage{algorithm}
\usepackage{xurl}           % must be loaded before hyperref
\usepackage{hyperref}
\usepackage{xr-hyper}
\newcommand{\rr}{\mathop{{\rm I}\mskip-4.0mu{\rm R}}\nolimits}
\newcommand{\Z}{\mathop{{\rm Z}\mskip-7.0mu{\rm Z}}\nolimits}

\newtheorem{lemma}{Lemma}
\newtheorem{proposition}{Proposition}

\newtheorem{remark}{\textbf{Remark}}
\newtheorem{assumption}{Assumption}
\newtheorem{definition}{Definition}

\def\tsc#1{\csdef{#1}{\textsc{\lowercase{#1}}\xspace}}
\tsc{WGM}
\tsc{QE}
\tsc{EP}
\tsc{PMS}
\tsc{BEC}
\tsc{DE}
\begin{document}
	\let\WriteBookmarks\relax
	\def\floatpagepagefraction{1}
	\def\textpagefraction{.001}
	
	\shorttitle{A Data-Driven Approach To Preserve Safety and Reference Tracking for Constrained CPS Under Network Attacks}
	\shortauthors{M. Attar et~al.}
	
	\title[mode=title]{A Data-Driven Approach To Preserve Safety and Reference Tracking for Constrained Cyber-Physical Systems Under Network Attacks}
	
	\tnotetext[1]{This work was supported in part by the Natural Sciences and Engineering Research Council of Canada (NSERC).}
	\tnotetext[2]{Mehran Attar and Walter Lucia are with the department of Cybersecurity and Intelligent Systems Engineering (CISE), Concordia University, Montreal, QC, H3G 1M8, Canada. \texttt{\small mehran.attar@concordia.ca}, \texttt{\small walter.lucia@concordia.ca}.}
	
	\author[1]{Mehran Attar}[orcid=0000-0002-4004-0846]
	\author[1]{Walter Lucia}[orcid=0000-0003-3776-8331]
	
	%	\ead{mehran.attar@concordia.ca}
	%	\ead{walter.lucia@concordia.ca}
	%	\ead[url]{https://www.concordia.ca/}
	
	\affiliation[1]{
		organization={Cybersecurity and Intelligent Systems Engineering (CISE) department
			, Concordia University},
		city={Montreal},
		state={Quebec},
		country={Canada}
	}
	
	%	\cortext[cor1]{Corresponding author}
	\cortext[cor2]{Corresponding author: \texttt{walter.lucia@concordia.ca}}
	
	\begin{abstract}
		This paper proposes a worst-case data-driven control architecture capable of ensuring the safety of constrained Cyber-Physical Systems under cyber-attacks while minimizing, whenever possible, potential degradation in tracking performance. To this end, a data-driven robust anomaly detector is designed to detect cyber-attack occurrences. Moreover, an add-on tracking supervisor module allows safe open-loop tracking control operations in case of unreliable measurements. On the plant side, a safety verification module and a local emergency controller are designed to manage severe attack scenarios that cannot be handled on the controller's side. These two modules resort to worst-case reachability and controllability data-driven arguments to detect potential unsafe scenarios and replace, whenever strictly needed, the tracking controller with emergency actions whose objective is to steer the plant's state trajectory in a predefined set of admissible and safe robust control invariant region until an attack-free scenario is restored. The effectiveness of the proposed solution has been shown through a simulation example.
	\end{abstract}
	
	%\includegraphics{figs/cas-grabs.pdf}
	
	%\item Research highlights item 1
	%\item Research highlights item 2
	%\item Research highlights item 3
	

	\maketitle
	
	%%%%%%%%%%%%%%% SECTION %%%%%%%%%%%%%%%%%%%%%%
	\section{Introduction}\label{sec:introduction}
	Cyber-Physical Systems (CPSs) are advanced engineering systems that closely integrate computation and communication technologies. Owing to their enhanced capabilities compared to traditional systems, CPSs have swiftly been adopted across various sectors, including water treatment, energy management, aerospace, and manufacturing. 
%	\todo[inline]{At the end of the paragraph in blue I see a missing refernece. Perhaps you need to upload the latest bib file?}
n this context, recent studies on networked control systems have investigated how network-induced limitations affect control performance~\cite{YANG2025112120, rsetam2025gpio} . Beyond such communication-efficiency challenges, the reliance of CPSs on network communications also introduces significant security and privacy concerns, particularly vulnerability to cyber-attacks
		%	However, their reliance on network communications also introduces significant security and privacy challenges, particularly vulnerability to cyber-attacks 
		\cite{cardenas2008secure, zha2025privacy}. Therefore, in the research community, a lot of attention has been given to developing control solutions to detect and identify attacks, mitigate their presence and preserve the plant’s safety and privacy, see, e.g., \cite{farokhi2019ensuring, attaractive, dibaji2019systems}. 
	Recently, there has been a growing trend in constrained CPS to explicitly address safety concerns. In \cite{gheitasi2021safety,zhang2020reachability}, by resorting to reachability analysis and set-theoretic concepts, authors have proposed a control architecture to preserve the safety of the plant. 
	The authors in \cite{franze2023cyber} presented an adaptive Model Predictive Control (MPC) framework capable of addressing security challenges in constrained networked control systems for various types of cyber-attacks. %The proposed method can be easily adapted to various types of attacks without altering the core resilient MPC structure. 
	%A key feature is the inclusion of detection units, which ensure a strong link between the domain of attraction of the MPC controller and the safe operating region of the system's state. 
	In \cite{franze2020resilience}, a distributed control architecture is proposed for discrete-time linear time-invariant multi-agent networked systems under replay attacks. 
	%The controller offers two key benefits: it can detect attacks using simple set-membership tests, and it confines the system's state trajectory to a safe region where constraints are met, and normal operations are maintained despite any allowable attack. This is achieved by applying model predictive control techniques, tailored to mitigate harmful system behaviors within a defined time frame until normal, attack-free communication is restored.
	% In \cite{franze2019resilient}, by exploiting a set-theoretic model predictive control
	% (MPC) framework, a resilient control architecture has been proposed for CPSs under replay attacks. 
	%
	In \cite{escudero2023safety}, a set-theoretic method has been developed to synthesize optimal LTI filters that constrain control inputs, preventing the reachability of unsafe states caused by resource-limited actuator or sensor attacks.	
	In \cite{wei2024resilient}, a
	distributed MPC and attack detection framework is proposed for constrained linear multi-agent systems under adversarial attacks.  In \cite{gheitasi2022worst}, a modular architecture capable of preserving the plant's safety while minimizing tracking performance loss has been designed for CPSs subject to state and input constraints.

	All the above mentioned contributions are derived assuming an accurate a-priori knowledge of the system's dynamical model. However, obtaining an accurate mathematical model can be challenging, especially when the system's behavior is influenced by unknown or uncertain factors such as disturbances \cite{yang2021scalable}.
	Consequently, recent efforts have focused on developing data-driven control architectures to safeguard the safety of constrained CPSs against cyber-attacks. For instance, in \cite{attar2024safety}, the authors propose a solution utilizing data-driven set-theoretic concepts to ensure the safety of a constrained CPS.  In \cite{gao2022resilient}, a resilient reinforcement learning approach has been developed to deal with partially linear systems subject to Denial-of-service (DoS) attacks. In \cite{liu2022data}, a data-driven MPC approach has been developed to compute input sequences and predicted outputs obtained from convex programming programs based on pre-collected input-output data. Building on this, a data-driven resilient controller is introduced, ensuring local input-to-state stability under specific denial-of-service attacks and noise levels.
	%
	%This approach does not require prior system knowledge or model identification.
	% In \cite{bajelani2024modular}, a modular data-driven safety filter for Cyber-Physical Systems has been designed to ensure safety by validating control commands regardless of their integrity. This filter operates independently of the control command's reliability and the anomaly detector's characteristics and can be integrated with resilient controllers to maintain high performance during normal operation and safety during attacks.
	%
	%%%%%%%%%%%%%%%% SECTION %%%%%%%%%%%%%%%%%%

	%%%%%%%%%%%%%%%% SECTION %%%%%%%%%%%%%%%%%%
	
	\subsection{Contributions} \label{sec:contributions}
	
	The above-described state-of-the-art show that most of the existing results on the data-driven design of safety-preserving control architectures neglect the tracking performance degradation problem under False Data Injection (FDI) attacks and, to the best of the author's knowledge, such a problem has not yet been explored in the literature.  
	%In particular, to the best of the author’s knowledge, the problem of preserving the safety and tracking performance of a constrained CPS  from a collection of noisy input-state trajectories has not yet been explored in the literature. 
	%
	Consequently, this paper goes in the direction of filling the existing gap, developing a solution that tries to minimize, whenever possible and in a data-driven fashion, the tracking performance loss under cyber-attacks. {In particular, we here develop a control architecture that leverages the data-driven results developed in  \cite{alanwar2021data, de2019formulas} to extend the model-based solution in \cite{gheitasi2022worst} and \cite{attar2024safety} to develop a novel data-driven control architecture that is capable of preserving the safety of the plant while minimizing, whenever possible, the tracking performance loss due to cyber-attacks on the communication channels. Consequently, the main novelty lies in integrating different data-driven modules (namely safety verification, emergency controller, and tracking supervisor) into a unified worst-case framework/architecture that preserves safety while reducing tracking degradation under arbitrary FDI attacks. With respect to the existing state-of-the-art, the main contributions of this work can be summarized as follows:}
	\begin{itemize}
		%	\item Unlike the approach in \cite{gheitasi2021safety}, the proposed solution is entirely data-driven.
		%
		%		\item While the contributions in \cite{gao2022resilient,liu2022data} focus on safety and resilience under DoS attacks of finite time and frequency, we here deal with unconstrained FDI attacks. 
		\item {While the contributions in \cite{gao2022resilient,liu2022data}
			focus on safety and resilience under DoS attacks of finite
			time and frequency, we here deal with unconstrained FDI
			attacks that can affect both measurement and actuation
			channels with arbitrary duration and frequency.}
		%%%%%%%%%%%%%%%%%%%%%%%%%%%%%%%%%%%%%%%%%%%%%%%%%%%%%%%%%%%
		%\item Existing data-driven solutions only focus on preserving the plant's safety under FDI attacks (see, e.g., \cite{attar2024safety, escudero2023safety, wabersich2021predictive}). On the other hand, to the best of the author's knowledge, the proposed solution defines one of the first approaches capable of ensuring safety while minimizing the tracking performance losses during FDI attacks.
		%%%%%%%%%%%%%%%%%%%%%%%%%%%%%%%%%%%%%%%%%%%%%%%%%%%%%%%%%%%%
		\item {Existing data-driven solutions under FDI attacks mainly focus on preserving the plant's safety; see, e.g.,
			\cite{escudero2023safety, attar2024safety, wabersich2021predictive}. On the
			other hand, to the best of the authors' knowledge, the
			proposed solution defines one of the first data-driven
			approaches capable of ensuring safety while explicitly
			minimizing the tracking performance loss during FDI
			attacks.}
		\item Differently from \cite{gheitasi2022worst}, we do not require that attacks be detected instantaneously. Consequently, a more general setup that considers attack detection delays is developed. {Technically, this leads to a new formulation
			of the robust state predictions performed by the tracking
			supervisor module, which allows supervised open-loop
			tracking from the last reliable measurement until such an
			operation is no longer deemed safe or beneficial for
			tracking.}
		\item The proposed data-driven emergency controller makes use of data-driven robust backward reachable sets, whose computation is not available in \cite{alanwar2021data}. Moreover, differently from  \cite{gheitasi2021safety}, robust backward reachable sets are not model-based and are defined on an augmented state-space description that addresses computational issues specific to the data-driven formulation (not of concern in \cite{gheitasi2021safety}).
		\item The proposed emergency controller extends the data-driven predictive controller introduced in \cite{attar2023data}. Particularly, unlike \cite{attar2023data}, this work incorporates the design of $L$ Voronoi partitions of $\mathcal{X}_{\eta}$, state feedback controllers with associated RCI regions, and a family of ROSC sets. These enhancements aim to mitigate the tracking performance loss compared to the method presented in \cite{attar2024safety}.
	\end{itemize}

	A Matlab implementation of the algorithm developed in this paper is available on the following GitHub page:
	\href{https://github.com/PreCyseGroup/data-driven-tracking-supervisor}{https://github.com/PreCyseGroup/data-driven-tracking-supervisor}.
	
	%A Matlab implementation of the algorithm developed in this paper is available on the following GitHub page: 
	%\href{https://github.com/PreCyseGroup/data-driven-tracking-supervisor}{\url{https://github.com/PreCyseGroup/data-driven-tracking-supervisor}}

	%%%%%%%%%%%%%%% SECTION %%%%%%%%%%%%%%%%%%%%%%
	\section{Preliminaries and problem formulation}\label{sec:preliminaries}
	
	Denote with $k\in \Z_+=\{0,1,\ldots\}$ a discrete-time index, and consider the discrete-time Linear Time-Invariant (LTI) system 
	\begin{equation}\label{LTI-for-definitions}
		z_{k+1}=\Phi z_k + G \mu_k + p_k,\quad z_k\in \mathcal{Z},\, \mu_k\in \mathcal{U}_{\mu},\, p_k\in \mathcal{P},
	\end{equation}
	where $p_k$ is an unknown but bounded process disturbance and
	$\mathcal{Z},$  $\mathcal{U}_{\mu},$ and $\mathcal{P}$
	are compact sets.

	{
		{\definition \label{def:RPI_set}
		A set $\mathcal{T} \subseteq \mathcal{Z}$ is called Robust Positive Invariant (RPI) for the autonomous system 
		$z_{k+1}=\Phi z_k + p_k$  if $\forall z \in \mathcal{T}, \Phi z  + p \in \mathcal{T}, \forall p \in \mathcal{P}$ \cite[Definition 11.20]{borrelli2017predictive}.}
	}
	
	{\definition \label{def:RCI_set}
		A set $\mathcal{T} \subseteq \mathcal{Z}$ is called Robust Control Invariant (RCI) for \eqref{LTI-for-definitions} if $\forall z \in \mathcal{T}, \exists \mu \in \mathcal{U}_{\mu}: \Phi z + G \mu + p \in \mathcal{T}, \forall p \in \mathcal{P}$ \cite[Definition 11.22]{borrelli2017predictive}.}

		{
		\begin{definition}\label{def:model_based_controllable_sets}
			Consider the LTI system \eqref{LTI-for-definitions}  and a target set $\mathcal{C}_i \subseteq \mathcal{Z}.$ The set of states $\mathcal{C}_{i+1} \subseteq \mathcal{Z}$ Robust One Step Controllable  (ROSC) to $\mathcal{C}_i$ is \cite{blanchini2008set}:
			\begin{equation}\label{eq:ROSC-set}
				\mathcal{C}_{i+1}\!=\! \{z \in \mathcal{Z}: \!\exists \mu \in \mathcal{U}_{\mu}\!:\! \Phi z + G \mu + p\in \mathcal{C}_i, \forall p \in  \mathcal{P} \}.
			\end{equation}
		\end{definition}
	}

%	{\definition \label{def:model_based_controllable_sets}
%		{\color{red}\cite{blanchini2008set}} Consider the LTI system \eqref{LTI-for-definitions} and a target set $\mathcal{C} \subseteq \mathcal{Z}$.
%		The set of states Robust One-Step Controllable (ROSC) to $\mathcal{C}$  is
%		%
%		\begin{equation}\label{eq:ROSC-set}
%			\mathcal{C}\!=\! \{z \in \mathcal{Z}: \!\exists \mu \in \mathcal{U}_{\mu}\!:\! \Phi z + G \mu + p\in \mathcal{C}, \forall p \in  \mathcal{P} \}.
%		\end{equation}}
%	%

{
\begin{definition}\label{def:RORS}
	Considering the LTI system  \eqref{LTI-for-definitions} and a set   $\mathcal{R}_i\subset \mathcal{Z}.$ The set of states $\mathcal{R}_{i+1}$ Robust One Step Reachable (RORS) from $\mathcal{R}_{i}$ is \cite[Section 11.3]{borrelli2017predictive}: 
		\begin{equation}\label{eq:RORS_model}
		\mathcal{R}_{i+1}\!=\!\{ z\!:\! \exists \bar{z}\in \mathcal{R}_i, \mu \in \mathcal{U}_{\mu},  p \!\in \! \mathcal{P}: z\! = \Phi \bar{z} + G \bar{\mu} + p \}
	\end{equation}
\end{definition}
}
	
%	{\color{red}
%		{\definition\label{def:RORS}
%			Consider the LTI system \eqref{LTI-for-definitions}, a state $\bar{z}\in \mathcal{Z}$ and a control input $\bar{\mu}\in \mathcal{U}_{\mu},$ the set of states
%			 Robust One-Step Reachable Set (RORS) $\mathcal{R}$ is defined as follows \cite[Section 11.3]{borrelli2017predictive}: 
%			\begin{equation}\label{eq:RORS_model}
%				\mathcal{R}\!=\!\{ z^+\!\!\!:\! \exists p \!\in \! \mathcal{P}\, s.t.\, z^+\! = \Phi \bar{z} + G \bar{\mu} + p \}.
%		\end{equation}}
%	}
	%
	%	{\definition \label{def:safety_def} At $k\geq 0$, the system $z_{k+1}=\Phi z_k + G \mu_k + p_k,$ subject to constraints is said \textit{safe} if $z_k\in \mathcal{Z},$ $\mu_k\in \mathcal{U}_{\mu}$ and $\left (\Phi z_k + G \mu_k \right )\oplus \mathcal{P}\subseteq \mathcal{Z};$ \textit{unsafe} otherwise.}
	%
	{
		\definition \label{def:safety_def} 
		At time $k\geq 0$, the state-control pair $(z_k,\mu_k)$ of system~\eqref{LTI-for-definitions} is said to be safe if
		$z_k\in\mathcal{Z}$, $\mu_k\in\mathcal{U}_{\mu}$, and
		$
		\Phi z_k+G\mu_k+p \in \mathcal{Z}, \, \forall p\in\mathcal{P}.
		$
		Otherwise, the pair $(z_k,\mu_k)$ is said to be unsafe.
	}
	
	Of interest for this paper are networked control systems setups where the plant and the tracking controller are spatially distributed and a communication medium is used to exchange state measurements and control inputs. By considering the possibility of cyber-attacks on the communication channels, we also assume that an anomaly detector is implemented local to the controller. In particular, the assumed plant model, controller, cyber-attack actions, and anomaly detector logic can be formalized as follows.
	%%%%%%%%%%%%%%%%%%%%%%%%%%%%%%%%%%%%%%%%
	\subsection{Plant model and safety constraints} \label{sec:plant_model}
	We consider plants whose dynamic evolution can be described by means of the following LTI system
	\begin{equation}\label{eq:linear_system}
		x_{k+1} = Ax_k + Bu_k + w_k,
	\end{equation}
	where $x_k \in \rr^n, u_k \in \rr^m$ and $A, B$ are the unknown system matrices and $w_k$ is a bounded disturbance that lies into a known compact set $\mathcal{W}\subset\rr^n.$ Due to physical limitations and safety reasons, the following set-membership constraints are prescribed:
	\begin{equation}
		x_k \in \mathcal{X}\subset \rr^n,\quad u_k \in \mathcal{U}\subset \rr^m,\quad  \label{eq:constraints}
	\end{equation}
	with $\mathcal{X}\subset \rr^n,\,\mathcal{U}\subset \rr^m, \mathcal{W} \subset \rr^n$ compact sets described by means of Zonotopes containing the origin.
	%%%%%
	
	\begin{assumption} \label{assumption:data_condition}
		The  matrices $A,$ $B$ of \eqref{eq:linear_system} are unknown.
		Moreover, a collection of $N_t>0$ noisy input-state trajectories is available,
		\begin{equation}\label{eq:available_trajectories}
			\left\{\left\{u^{(i)}_k\right\}^{N_{s}^{(i)}-1}_{k=0},\,  \left\{x^{(i)}_k\right\}^{N_{s}^{(i)}-1}_{k=0}\right\}_{i=1}^{N_t},
		\end{equation}
		where $N_{s}^{(i)}>0$ is the number of samples in each trajectory {and $N_D := \sum_{i=1}^{N_t} N_s^{(i)}$}. 
		By arranging the collected data into two matrices { $X_{-}\in \rr^{n \times N_D},$ $U_{-}\in \rr^{m \times N_D},$ }where 
		\begin{eqnarray}
			\!\!\!\!X_{-}\!\!\!\!\!\!&=&\!\!\!\!\!\! \left[
			x^{(1)}_0, \cdots, x^{(1)}_{N_{s}^{(1)}-1}, \cdots, x^{(N_t)}_0, \cdots, x^{(N_t)}_{N_{s}^{(N_t)}-1}\right], \\
			\!\!\!\!U_{-}\!\!\!\!\!\!&=&\!\!\!\!\!\!\left [
			u^{(1)}_0, \cdots, u^{(1)}_{N_{s}^{(1)} -1}, \cdots, u^{(N_t)}_0, \cdots, u^{(N_t)}_{N_{s}^{(N_t)} -1}
			\right ]. 
		\end{eqnarray}
		We assume that the matrix $\begin{bmatrix}
			X_-^T & U_-^T
		\end{bmatrix}^T$ has full row rank:
		\begin{eqnarray}\label{eq:rankCondition}
			&& \text{rank}(\begin{bmatrix}
				X_-^T & U_-^T
			\end{bmatrix}^T)=n+m \label{eq:rank_condition}
		\end{eqnarray}
		\hfill $\Box$
	\end{assumption}
	\begin{remark}
		\label{remark:data_driven_representation}Condition~\eqref{eq:rank_condition} ensures that the collected data \eqref{eq:available_trajectories} have been obtained for sufficiently persistent exciting input sequences \cite{de2019formulas}.
		In what follows, the one-step ahead shifted representation of $X_-$  is denoted as { $X_{+}\in \rr^{n \times N_D},$} where
		\begin{eqnarray}
			X_{+} \!\!\!\!\!\!&=&\!\!\!\!\!\!\left [
			x^{(1)}_1, \cdots, x^{(1)}_{N_{s}^{(1)}}, \cdots, x^{(N_t)}_1, \cdots, x^{(N_t)}_{N_{s}^{(N_t)}}\right ].
		\end{eqnarray}
	\end{remark}
	%%%%%%%%%%%%%%%%%%%%%%%%%%%%%%%%%%
	\subsection{Tracking controller}\label{sec:TC_controller}
	
	{The plant's state  $x_k$ is required to track a reference signal $r,$ with $r_k\in \rr^n$  constrained in a compact set $\mathcal{R},$ i.e. $r_k\in \mathcal{R} \subset  \rr^n,$ $\forall\,k\geq 0,$}  Moreover, it is assumed that a data-driven networked  tracking controller, obtained using the available data \eqref{eq:available_trajectories},  is available and it is described by the following model:
	\begin{equation} \label{eq:tracking_controller_logic}
		u_k = \eta(x_k,r_k), 
	\end{equation}
	where  
	$\eta(\cdot,\cdot): {\rr^n \times \rr^n }\rightarrow \rr^m$ is the networked tracking controller logic.
	\begin{assumption} \label{assumption:control_center} The networked tracking controller \eqref{eq:tracking_controller_logic} is given, and in the absence of attacks, it ensures that the plant's constraints \eqref{eq:constraints} are satisfied regardless of any realization of the admissible disturbance $\mathcal{W}$ and any admissible bounded reference signal.  For example, such a controller can be designed resorting to the solutions developed in \cite{de2019formulas, 6339058, 7932940}. {he resulting closed-loop system is robust positive invariant
		%	 robust invariant set associated with the 
		%	 closed-loop tracking system induced by the controller $\eta(\cdot,\cdot)$ is assumed to be available and hereafter denoted by
			 is a known or estimated set 
			  $\mathcal{X}_{\eta}\subseteq \mathcal{X}$ for any $r_k \in \mathcal{R}$ and $w_k \in \mathcal{W},$  for any $k\geq 0.$  
%			\todo[inline]{Mehran, is there a citation we can use for Certified robust invariant region? if yes, please add the citation in the following sentence and also in the related comment in the letter (R4-9)}
			The set $\mathcal{X}_{\eta}$ should be interpreted  as a data-driven certified admissible robust invariant region for the closed-loop system \cite{RCI_LPV20232}.}
		
%		The largest RCI set associated with the tracking controller is hereafter denoted by $\mathcal{X}_{\eta} \subseteq\mathcal{X}.$ 
	\end{assumption}
	%%%%%%%%%%%%%%%%%%%%%%%%%%%%%%%%%%%%%%%%%%%%%%%
	\subsection{Networked cyber-attacks}\label{sec:system_under_attack}
	
	We assume that the communication channels between the plant and the controller are vulnerable to FDI attacks. Consequently, the closed-loop evolution of \eqref{eq:linear_system} under FDI attacks on both the actuation and measurement channel is 
	\begin{equation}\label{eq:system_under_attack}
		x_{k+1} = Ax_k + Bu'_k + w_k, \quad u_k = \eta(x'_k,r_k), 
	\end{equation}
	where $u'_k:=u_k + u^a_k, x'_k:=x_k + x^a_k$, with $u^a_k\in \rr^m$ and $x^a_k\in \rr^n$ the vectors injected by the attacker. Moreover, the attacker's injections can have any frequency and duration.
	%%%%%%%%%%%%%%% SECTION %%%%%%%%%%%%%%%%%%%%%
	\subsection{Anomaly detector} \label{sec:anomaly_detector}
	
	We assume that a passive data-driven binary anomaly detector is available on the controller's side to detect the presence of cyber-attacks by leveraging the received state measurements $\{x'_{t}\}_{t=0}^k$ and computed control inputs $\{u_{t}\}_{t=0}^k$. The following model abstractly describes the anomaly detector
	\begin{equation} \label{eq:anomaly_detector}
		d_k = \text{\textit{D}}\left(\{x'_{t}\}_{t=0}^k, \{u_{t}\}_{t=0}^{k-1}, \mathcal{W}\right),
	\end{equation}
	where $D(\cdot,\cdot, \cdot)$ is the detection logic and $d_k\in[0, 1].$ In what follows, we will assume that $d_k=1$ denotes the presence of an anomaly. Moreover, it is assumed that the detection mechanism is capable of providing attack detection with an estimated bounded delay $0 \leq \tau < \infty$  \cite{shi2017causality,kurt2018distributed} that we assume is estimated experimentally, evaluating the performance of \eqref{eq:anomaly_detector} for different attack scenarios.
	{Note that the anomaly detector \eqref{eq:anomaly_detector} does not ensure the absence of undetectable attacks.   However,  this is not a limitation of the used detector but a known issue; it has been proved in  \cite{smith2015covert} that, for example, covert attacks cannot be detected by any actions taken only on the controller's side. 	Consequently, additional actions must be taken on the plant's side to ensure the detection of any class of intelligent coordinated attacks.  In particular, such a problem is hereafter addressed in Section~\ref{sec:safety_verification} , when the safety verification module is introduced.}
	
	%{\color{red}It is worth noting that the anomaly detector is not assumed to detect all possible cyber-attacks; rather, it provides a first detection layer for detectable anomalies, while the plant-side Safety Verification module is responsible for preserving safety even when some attacks remain undetected.}
	
	% 	%
	\subsection{Problem Formulation} \label{sec:problem_formulation}
	
	In the considered networked architecture, we assume that only the networked controller is aware of the reference signal $r_k$. Therefore, %since $r_k$ is
	%unknown on the plant side, 
	it is acceptable to experience performance degradation during cyber-attacks as long as the plant’s safety and recovery (after the attack) are guaranteed. However, it is also desirable to design the control architecture to minimize the tracking performance degradation due to cyber-attacks. Consequently, the problem of interest can be stated as follows:

	{\it Under Assumptions \ref{assumption:data_condition}-\ref{assumption:control_center}, design a data-driven control architecture for the constrained system  \eqref{eq:linear_system}- \eqref{eq:constraints}  ensuring (i) plant's safety and tracking performance loss minimization during any cyber-attack, and (ii) performance recovery in the post-attack phase.}
	%
	
	%%%%%%%%%%%%% SECTION %%%%%%%%%%%%%%%%%%%%%%%
	
	\section{Proposed Solution} \label{sec:proposed_solution}
	
	To design the proposed data-driven solution, we consider a worst-case scenario in which a cyber-attack can affect both communication channels, either simultaneously or at different times. 
	If only safety is of interest, then a data-driven solution to the problem can be obtained by adapting the strategy developed in \cite{attar2024safety}. Indeed, upon detecting an attack, such a strategy prescribes disconnecting the networked tracking controller and activating a local emergency controller for safety-preserving purposes. The drawback of such a solution is that it jeopardizes the tracking task regardless of the nature of the attack, producing an overly conservative mitigation strategy that ultimately results in poor tracking performance.
	On the other hand, in the here proposed solution, we extend the approach in \cite{attar2024safety}
	with the aim of minimizing, whenever possible,  the tracking performance loss caused by cyber-attacks. In particular, we argue that tracking performance degradation can be minimized by implementing targeted measures to counteract attacks on the actuation and measurement channels. To this end, we enhance the networked control scheme shown in Fig.~\ref{fig:proposed_architecture} with four additional modules:
	\begin{itemize}
		\item A Safety Verification (\textit{SV}) subsystem, local to the plant, whose  objective is to prevent the plant from reaching unsafe configurations; {This module makes use of RORS (as in Definition~\ref{def:RORS}) to check for potential reachable unsafe conditions (as in Definition~\ref{def:safety_def})}. 
		\item An Emergency Controller (\textit{EC}), local to the plant, and it is used when the control signal $u'_k$ is deemed untrustworthy by the safety verification module. {The EC is built using the set-theoretic concepts of RCI set (as in Definition \ref{def:RCI_set})  and ROSC set (as in Definition~\ref{def:model_based_controllable_sets})  to confine, in a finite number of steps, the trajectory on a chosen safe equilibrium point}
		% in charge of guaranteeing the plant safety when the received control signal $u^{'}_k$ is untrustworthy.
		%
		\item An Anomaly Detector (\textit{D}) in charge of detecting the presence of attacks. 
		\item A Tracking Supervisor (\textit{TS}), local to the tracking controller, responsible for minimizing the tracking performance loss while ensuring safety when $x^{'}_k$ is invalid. This module is activated when the Anomaly Detector detects the presence of a cyber-attack (i.e., if  $d_k=1$) {and it will make use of a tracking performance index based on the overlap of computed RORS sets and a Voronoi partition of the admissible state-space.}
	\end{itemize}
	\begin{figure}[htbp]
		\centering
		\includegraphics[width=\columnwidth]{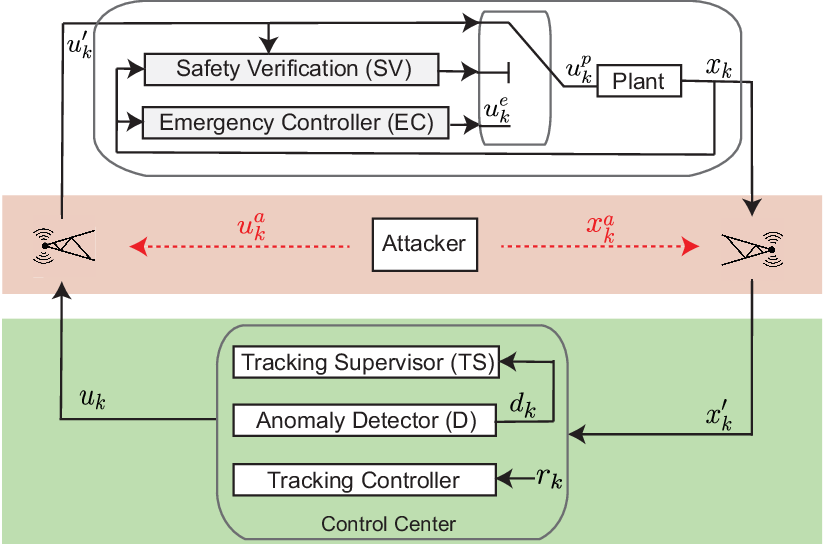}
		\caption{Proposed Control Architecture}
		\label{fig:proposed_architecture}
	\end{figure}

	\begin{remark}\label{remarl:justify_architecture}
		In various CPS application domains (such as the Smart Grid), the networked controller manages and coordinates multiple subsystems that collaborate toward a common objective. In this context, the networked controller is uniquely aware of their reference signals $r_k$. Thus, the state-feedback emergency controller shown in Fig.~\ref{fig:proposed_architecture} cannot replace the remote tracking controller; instead, whenever used, it temporarily suspends tracking to preserve safety.
	\end{remark}
	%
	%%%%%%%%%%%%%%% SECTION %%%%%%%%%%%%%%%%%%%%
	\subsection{Safety Verification Module} \label{sec:safety_verification} 
	
	This module aims to prevent the plant from reaching unsafe configurations (see Definition~\ref{def:safety_def}). Consequently, it complements the anomaly detection rules \eqref{eq:anomaly_detector}, and it ensures that any attack is detected at least one step before it could affect the plant's safety.  In particular, given the received $u_k'$ and the local available measurement of $x_k,$ the safety module checks the following possible anomalies: 
	\begin{equation} \label{eq:safety_one_step_ev}
		u_k'\notin \mathcal{U},  \quad  \quad \mathcal{S}^{+}_k \not \subseteq \mathcal{X}_{\eta},
	\end{equation}
	where  $\mathcal{S}^{+}_k$ denotes the RORS set starting from $x_k$ and $u_k'.$ Therefore, $\mathcal{S}^{+}_k \not \subseteq \mathcal{X}_{\eta}$ represents a situation where exists a disturbance realization such that $x_{k+1}$ does not fulfill the constraints. Consequently, the safety verification logic can be summarized as follows
	\begin{itemize}
		\item \textbf{If} $u_k'\notin \mathcal{U}$ or $\mathcal{S}^{+}_k \not \subseteq \mathcal{X}_{\eta}, $
		\textbf{then} the $u_k'$ has been corrupted, and the \textit{EC} is activated.
		\item \textbf{Else} the  $u_k'$ is deemed safe and applied to the plant. 
	\end{itemize}

	If the model of the plant is known then $\mathcal{S}_k^+:=Ax_k \oplus Bu'_k \oplus \mathcal{W}$. However, based on Assumption~\ref{assumption:data_condition}, the system matrices $A, B$ are unknown. Therefore, the forward one-step evolution of the system cannot be directly derived from \eqref{eq:RORS_model}. Thus, in what follows,  by adopting the data-driven solution developed in \cite{ alanwar2021data,koch2021provably}, we compute an outer approximation of $\mathcal{S}^{+}_k,$ namely $\hat{\mathcal{S}}^{+}_k,$ such that $\hat{\mathcal{S}}^{+}_k \supseteq \mathcal{S}^{+}_k.$  For the sake of completeness, the following two lemmas summarize how $\hat{\mathcal{S}}^{+}_k$ can be computed in a data-driven fashion.
	\begin{lemma}{}\label{lemma:noise_zonotope} \it \cite[Lemma 1]{alanwar2021data}
		Let $T =\displaystyle \sum_{i=1}^{N_t} N^{(i)}_s$ and consider the following concatenation of multiple noise zonotopes 
		$$\mathcal{M}_w=\mathcal{M}_w(C_w, [G^{(1)}_{{M}_w}, \ldots, G^{(qT)}_{{M}_w}]),
		$$ 
		where
		$C_w\in \rr^{n\times (n+m)}=[c_{w},\ldots,\,c_w]$, and 
		{$G_{M_w}\!\!\!\!\!\!:=\!\!\!\!\!\left[
		G_{M_w}^{(1)},\ldots,G_{M_w}^{(qT)}
		\right]\!\!\in\!\! \rr^{n\times T(n+m)}$} is built $\forall\,i \in \{ 1, \ldots, q\},\, \forall\, j \in \{2, \ldots, T-1\}$ as 
		\begin{equation}
			\begin{array}{rcl}
				G^{(1+(i-1)T)}_{{M}_w} &=& \begin{bmatrix}
					g^{(i)}_{w} & 0_{n \times (T-1)}
				\end{bmatrix},
				\\
				G^{(j+(i-1)T)}_{{M}_w} &=& \begin{bmatrix}
					0_{n \times (j-1)} & g^{(i)}_{w} & 0_{n \times (T-j)}
				\end{bmatrix},
				\\
				G^{(T+(i-1)T)}_{{M}_w} &=& \begin{bmatrix}
					0_{n \times (T-1)} & g^{(i)}_{w}
				\end{bmatrix}.
			\end{array}
		\end{equation}
		Then, the matrix zonotope
		\begin{equation}\label{eq:compute_Mzono_AB}
			\begin{array}{rcl}
				
				\mathcal{M}_{{A} {B}}\!\!\!\!\!\!&=&\!\!\!\!\!\! (X_+ - \mathcal{M}_w) \begin{bmatrix}
					X_- \\ U_-
				\end{bmatrix}^{\dagger}   
				\\
				\!\!\!& := &\!\!\!\!\!\! \{[\hat{A},\, \hat{B}]: 
				[\hat{A},\, \hat{B}]\! =\! C_{AB} + \displaystyle\sum_{i=1}^{T} \beta^{(i)}G^{(i)}_{M_{AB}},\\
				&&
				-1\leq \beta^{(i)} \leq 1 \},
			\end{array}
		\end{equation}
		where $\left [ \cdot \right]^{\dagger}$ is the right pseudo inverse operator and
		$$
		\begin{array}{c}
			C_{AB} = (X_{+}-C_{w})\left(
			[X^T_{-},\, U^T_{-}]    ^{T}\right)^{\dagger},\\
			G_{M_{AB}} =\left[
			G^{(1)}_{M_w}
			\left(
			[X^T_{-},\, U^T_{-}]    ^{T}\right)^{\dagger}\!\!, \ldots, G^{(qT)}_{M_w}\left(
			[X^T_{-},\, U^T_{-}]    ^{T}\right)^{\dagger}
			\right],
		\end{array}
		$$
		contains the set of all system matrices $[\hat{A},\, \hat{B}]$ consistent with \eqref{eq:available_trajectories} and {$\mathcal{W}$} and such that $[A,B]\in \mathcal{M}_{AB}.$ $\hfill\square$
	\end{lemma}
	\begin{lemma} (Adapted from \cite[Theorem 1]{alanwar2021data})\label{lemma:over_approx_forward_sets}
		\it 
		The set $\hat{\mathcal{S}}^+_k\subset \rr^n,$ computed as 
		\begin{equation}
			\label{eq:over_approximation_of_one_step_evolution}
			\hat{\mathcal{S}}^+_k = \mathcal{M}_{AB}[x_k,u'_k]^T \oplus \mathcal{W},  
		\end{equation}
		is a conservative outer approximation of ${\mathcal{S}}^+_k$ and $\oplus$  is the Minkowski set sum operator. 
	\end{lemma}
	%
	%%%%%%%%%%%%%%% SECTION %%%%%%%%%%%%%%%%%%%%
	\subsection{Emergency Controller Module} \label{sec:emergency_controller} 
	Since the Emergency Controller does not have access to the reference signal $r_k,$ its objective is to preserve the plant's safety (during the attack) and to guarantee performance recovery (when the attack is terminated). By denoting the logic of the emergency controller as
	\begin{equation}\label{eq:emergency_controller_law}
		u_k^{e}=f_e(x_k),\quad f_e: \mathcal{X}_e \subset \rr^n \rightarrow \mathcal{U}_e \subset \rr^m, 
	\end{equation}
	\eqref{eq:emergency_controller_law} must be designed to fulfill the following requirements:
	\begin{enumerate}
		\item Domain of attraction: 
		\begin{equation} \label{eq:conditions_emergency_controller1}
			\mathcal{X}  \supseteq \mathcal{X}_e\supseteq \mathcal{X}_{\eta}, \quad \mathcal{U}_e\subseteq \mathcal{U},
		\end{equation}
		\item Finite-time Uniformly Ultimately Bounded (UUB) stability in a set
		\begin{equation}\label{eq:conditions_emergency_controller2}
			\hat{\mathcal{T}}_0 \subseteq \mathcal{X}_{\eta}.
		\end{equation}
	\end{enumerate} 
	Condition \eqref{eq:conditions_emergency_controller1}  ensures that the emergency controller fulfills all the constraints and that it can be activated anytime and from any state reachable under the tracking controller, while condition \eqref{eq:conditions_emergency_controller2} guarantees that, in the post-attack phase, the networked controller can be safely re-activated in a finite number of steps. 
	%{\color{red} It is important to note that condition~\eqref{eq:conditions_emergency_controller2} does not imply recovery from an already unsafe state; rather, it guarantees finite-time convergence from the admissible/recoverable domain to a terminal safe region contained in \(\mathcal{X}_{\eta}\), from which the networked tracking controller can be safely re-activated.}
	
	%
	As shown in, e.g.,  \cite{attar2024safety}, a controller satisfying the requirement above can be obtained by customizing (starting from an equilibrium point inside $\mathcal{X}_{\eta}$), the set-theoretic data-driven controller developed in \cite{attar2023data}.   Although effective, such a design completely ignores the current state vector of the plant. Consequently, the action of the emergency controller can take the state trajectory far from the desired reference signal. In what follows, we mitigate such a drawback by designing a data-driven set-theoretic emergency controller from a set of $L \geq 1$ admissible equilibrium points $(x^l_e, u^l_e), l \in \mathcal{L}:=\{1, \cdots, L\}.$ A pair $(x^l_e, u^l_e)$ is considered admissible/safe if $x^l_e \in \mathcal{X}_{\eta}, \,\forall l$   and there exists a feedback controller
	\begin{equation} \label{eq:terminal_controller}
		u^l_k = K_l(x_k - x^l_e) + u^l_e,
	\end{equation}
	with gain  $K_l\in \rr^{m \times n},$ and
	such that {the associated smallest RPI set for the closed-loop system under the feedback controller $K_l$} (computed, e.g., using the data-driven methods developed in \cite{chen2021data, RCI_LPV20232} and \cite[Remark~5]{attar2023data}), namely $\hat{\mathcal{T}}^l_0 \in \rr^n$, is contained in the tracking controller's domain, i.e., $\hat{\mathcal{T}}^l_0 \subseteq \mathcal{X}_{\eta}$. 

	Note that by construction, the $L$ controllers \eqref{eq:terminal_controller} and associated RCI sets $\hat{\mathcal{T}}^l_0, l \in \mathcal{L}$ might not guarantee that the requirements \eqref{eq:conditions_emergency_controller1},\eqref{eq:conditions_emergency_controller2} are fulfilled, i.e., it might exist $x\in \mathcal{X}_{\eta}$ such that $ x \not \in {\bigcup^L_{l=1} \hat{\mathcal{T}}^l_0}=\mathcal{X}_e.$ Therefore, to enlarge the domain $\mathcal{X}_{e}$ and comply
	with \eqref{eq:conditions_emergency_controller1}, \eqref{eq:conditions_emergency_controller2} the following strategy is adopted.
	\begin{itemize}
		\item First, a Voronoi partition of $\mathcal{X}_{\eta}$ is created (see, e.g., the 5 partitions in Fig.~\ref{fig:trajectories}). Therefore, a family of polyhedral regions $\{\mathcal{V}_l \}^L_{l=1}$ enjoying the following properties is obtained:
		\begin{equation} \label{eq:voronoi_partition_condition}
			\mathcal{V}_l \!=\! \left \{ x \!\in\! \mathcal{X}_{\eta}: \! ||x - x^l_e||_2 \!\leq \! ||x - x^j_e||_2, \forall j \neq l, j \!\in\! \mathcal{L} \right \},
		\end{equation}
		\begin{equation}
			\bigcup_{l=1}^{L} \mathcal{V}_l = \mathcal{X}_{\eta}.
		\end{equation}
		\item 	Then, by resorting to a data-driven set-theoretic predictive controller paradigm proposed in \cite{attar2023data}, we enlarge the Domain of Attraction (DoA) of each $l-th$ controller \eqref{eq:terminal_controller} to
		cover the associated Voronoi partition $\mathcal{V}_l.$ To this end, a family of data-driven robustly controllable sets is recursively built by adapting the definition of ROSC sets (see Definition.~\ref{def:model_based_controllable_sets}) to the one-step evolution of model \eqref{eq:linear_system} under \eqref{eq:constraints} in an offline phase.  	In particular, families $\left\{\hat{\mathcal{C}}^l_j \right\}^{N_l}_{j=1},$  of $N_l \!\geq \!0$ of ROSC sets are built, with $N_l$ satisfying
		the termination condition $	\bigcup^{N_l}_{j=1}\left\{\hat{\mathcal{C}}^l_j \right\} = \mathcal{V}_l, \forall l=1, \cdots, L.$
	\end{itemize}

	As long as the data-driven computation of	 $\{\hat{\mathcal{C}}^l_j\}^N_{j=1},$ is concerned, we resort to the procedure, summarized in  Lemma~\ref{lemma:augmented_disc}, which resorts to an augmented description of the ROSC sets.  
	\begin{lemma}\label{lemma:augmented_disc} 
		\it Consider a collection of input-state trajectories for \eqref{eq:linear_system}-\eqref{eq:constraints} fulfilling Assumption~\ref{assumption:data_condition}. An inner approximation of the ROSC set \ref{def:model_based_controllable_sets} can be computed as follows \cite[Sec.~III.C]{attar2023data}:
		\begin{equation}\label{eq:inner_ROSC_data_driven_augm}
			\begin{array}{c}
				\hat{\mathcal{C}}^{l}_j= Proj_x(\hat{\Xi}^l_j)=\left\{\!x\! \in\! \rr^n\!:\! H_{\hat{\mathcal{C}}^{l}_j} x \leq h_{\hat{\mathcal{C}}^{l}_j}  \!\right\}, \\
				\hat{\Xi}^{l}_j \!=\text{In}_z\left\{\hat{\Xi}^{l}_{AB}\right\},
				\\
				\!\hat{\Xi}^{l}_{j_{AB}} \!= \!\!\!\!\!\! \displaystyle  
				\bigcap_{[\hat{A}_i,\hat{B}_i]\in \mathcal{V}_{AB}}\!\!\!\!\!\!\!\!\! 
				\left\{z=[x^T,u^T]^T\! \in \rr^{n+m}\!:\!H_{i}^lz\leq h_{i}^l\right\}, 
			\end{array}
		\end{equation}
		where 
		$\hat{\Xi}^{l}_{j_{AB}}$ is the $(x,u)-$ augmented  description of  the ROSC set,  $\mathcal{V}_{AB}$  denotes the  matrix vertices of $\mathcal{M}_{AB},$  $\text{In}_z(\cdot)$ is an operator computing a zonotopic inner approximation of a polytope, 
		$Proj_x(\hat{\Xi}^l_j)$ performs a projection of  $\hat{\Xi}^l_j$ into the $x-$domain, and 
		\begin{equation}\label{eq:H-rep_extended}
			H_i^l= \begin{bmatrix}
				H_x & 0 \\
				H_{\hat{\mathcal{C}}^l_{j-1}}\hat{A}_i & H_{\hat{\mathcal{C}}^l_{j-1}}\hat{B}_i\\
				0 & H_u
			\end{bmatrix}, 
			\quad
			h_z^i=\begin{bmatrix}
				h_x\\
				\Tilde{h}_{\hat{\mathcal{C}}^l_{j-1}}\\
				h_u
			\end{bmatrix}, 
		\end{equation}
		with 
		\begin{equation}\label{eq:compute_tilde_h}
			[\tilde{h}_{\hat{\mathcal{C}}^l_{j-1}}]_r = \min_{w\in \mathcal{W}}\left\{[h_{\hat{\mathcal{C}}^l_{j-1}}]_r - [H_{{\hat{\mathcal{C}}^l_{j-1}}}]_r w \right\},  
		\end{equation}
		and $[h_{\hat{\mathcal{C}}^l_{j-1}}]_r,$ $[H_{{\hat{\mathcal{C}}^l_{j-1}}}]_r$ the $r-th$ row of $h_{\hat{\mathcal{C}}^l_{j-1}}$ and $H_{{\hat{\mathcal{C}}^{l}}_{j-1}}.$
	\end{lemma} $\hfill\square$
	Given $\{\hat{\mathcal{C}}^l_j\}^N_{j=0},$ and a convex cost function $J(x_k,u),$ the online operations of the Safety Verification Module and Emergency data-driven set-theoretic controller can be summarized as in Algorithm~\ref{algorithm:set_theoretic_MPC_and_SV}.
	\begin{algorithm}[h!] 
		\textit{Offline:} Compute 
		$\hat{\mathcal{C}}_0^l$ and $\{\hat{\Xi}^{l}_j$, $\hat{\mathcal{C}}^{l}_j\}_{j=1}^{N_l}, l\!\!=\!\!1, \ldots, L$ as in \eqref{eq:inner_ROSC_data_driven_augm}. Set $f=1.$

		\noindent
		\xrfill[0.7ex]{1pt} 
		
		\textit{Online ($\forall\, k$):}
		
		\begin{algorithmic}[1]
			
			\If{ $f==0$ \textbf{or} $u_k'\notin \mathcal{U}$ \textbf{or} $\mathcal{S}^{+}_k \not \subseteq \mathcal{X}_{\eta}$}
			\Comment{\textit{Activate EC}}
			\State \Comment{\textit{EC starts}}
			\State Set $f=0$ and find $\bar{l} \in \mathcal{L}$ such that $x_k \in \mathcal{V}_{\bar{l}}$ \;
			\State Find $\bar{j}_k := \displaystyle \!\!\!\!\!\!\! \min_{j\in \{0,\ldots,N_l\}} \{j: x_k \in  \hat{\mathcal{C}}^l_j \}$\;
			\If{$\bar{j}_k==0$}
			\State { $u^{e}_k = f_0^l(x_k)$}, $f=1$
			\Else 
			\begin{equation}\label{eq:control_based_extended_control_regions}
				% \begin{array}{c}
					\displaystyle		u^{e}_k = \arg\min_u J(x_k, u) \quad s.t. 
					\left[x_k^T,u^T\right]^T \in \hat{\Xi}_{j_k}^l.
					% \end{array}
			\end{equation}
			%}
		\EndIf
		\Comment{\textit{EC ends}}
		\State Apply $u_k^e$  \Comment{\textit{EC} \textit{control law}}
		\Else
		\State Apply $u_k'$  \Comment{\textit{TS} \textit{control law}}
		\EndIf
	\end{algorithmic}
	\caption{Safety Verification (SV) and Emergency Controller (EC)}
	\label{algorithm:set_theoretic_MPC_and_SV}
\end{algorithm}
where the flag $f$ is used to make sure that the emergency controller, once activated, will remain active at least until the terminal region is reached.  Given the operations of the \textit{SV} and \textit{EC} modules (see Algorithm~\ref{algorithm:set_theoretic_MPC_and_SV}), the following proposition holds true (adapted from the results in \cite[Proposition~1]{gheitasi2022worst}):

\begin{proposition}\cite{gheitasi2022worst}{\label{proposition:em_safety_and_recovery}}
	Consider the sets of equilibrium pairs $\left\{ (x^l_e, u^l_e)\right\}$ and Voronoi partition $\left\{ \mathcal{V}_l\right\}^L_{l=1}$, the RCI sets $\left\{ \mathcal{T}_0^l\right\}^L_{l=1},$ and the families of ROSC sets $\left\{\hat{\mathcal{C}}_j^l \right\}_{j=1}^{N_l}, l\in \mathcal{L}$ computed according to \eqref{eq:inner_ROSC_data_driven_augm}. Then, if at $k=k'$ a persistent cyber-attack starts on the actuation channel and $x_{k'}\in \mathcal{V}_l, 1 \leq l \leq L,$ then the emergency controller ensures that safety and recovery are guaranteed (see Definition~\ref{def:safety_def}). Moreover, for $k \geq k' + N_l$ the tracking error $e_k=x_k - r_k$ is such that $e_k\leq d^{sup}(\hat{\mathcal{T}}^{l_i}_0, r_k),\forall k \geq k' + N_l$.
\end{proposition}
\begin{proof}
	{
	Assume that at time $k=k'$ a persistent cyber-attack starts on the
	actuation channel and that $x_{k'}\in\mathcal{V}_{l}$ for some
	$l\in\mathcal{L}$. According to Algorithm~\ref{algorithm:set_theoretic_MPC_and_SV}, once the
	Safety Verification module detects that the received input is not safe,
	the Emergency Controller is activated and the flag is set to $f=0$.
	Therefore, the Emergency Controller remains active until the terminal
	region associated with the selected Voronoi partition is reached.
	
	By construction, for each Voronoi region $\mathcal{V}_{l}$, a terminal
	RCI set $\hat{\mathcal{T}}^{l}_{0}\subseteq\mathcal{X}_{\eta}$ and a
	finite family of ROSC sets $\{\hat{\mathcal{C}}^{l}_{j}\}_{j=1}^{N_l}$
	are computed such that
	\begin{equation}
		\mathcal{V}_{l}\subseteq \bigcup_{j=0}^{N_l}\hat{\mathcal{C}}^{l}_{j},
		\qquad
		\hat{\mathcal{C}}^{l}_{0}:=\hat{\mathcal{T}}^{l}_{0}.
	\end{equation}
	Hence, since $x_{k'}\in\mathcal{V}_{l}$, there exists an index
	$j_{k'}\in\{0,\ldots,N_l\}$ such that
	$x_{k'}\in\hat{\mathcal{C}}^{l}_{j_{k'}}$.
	
	If $j_{k'}=0$, then $x_{k'}\in\hat{\mathcal{T}}^{l}_{0}$. Since
	$\hat{\mathcal{T}}^{l}_{0}$ is an RCI set for the local controller
	$f^{l}_{0}(\cdot)$, there exists an admissible input satisfying the input
	constraints such that the successor state remains in
	$\hat{\mathcal{T}}^{l}_{0}$ for all admissible disturbances. Therefore,
	the state and input constraints are satisfied and the plant remains safe.
	
	Consider now the case $j_{k'}>0$. From the definition of the ROSC set,
	for every $x\in\hat{\mathcal{C}}^{l}_{j}$, $j>0$, there exists an
	admissible control input $u\in\mathcal{U}$ such that the successor state (i.e., the next state $x_{k'+1}$)
	belongs to $\hat{\mathcal{C}}^{l}_{j-1}$ for all admissible disturbances.
	Equivalently, the augmented set
	$\hat{\Xi}^{l}_{j}$ contains admissible state-input pairs
	$(x,u)$ that robustly steer the state from
	$\hat{\mathcal{C}}^{l}_{j}$ to $\hat{\mathcal{C}}^{l}_{j-1}$ in one
	step. Therefore, the optimization problem in
	\eqref{eq:control_based_extended_control_regions} is feasible for
	$x_k\in\hat{\mathcal{C}}^{l}_{j}$ and returns an admissible input
	$u^{e}_{k}\in\mathcal{U}$ that satisfies the state and input constraints
	and guarantees $x_{k+1}\in\hat{\mathcal{C}}^{l}_{j-1},
	 \forall w_k\in\mathcal{W}.$
	By applying this argument recursively, the index of the active ROSC set
	decreases at least by one at each time step. Hence, after at most
	$j_{k'}\leq N_l$ steps, the state reaches the terminal RCI set
	$\hat{\mathcal{T}}^{l}_{0}$.
	
	Once $x_k\in\hat{\mathcal{T}}^{l}_{0}$, the local feedback controller
	$f^{l}_{0}(\cdot)$ is applied. Since $\hat{\mathcal{T}}^{l}_{0}$ is RCI
	and satisfies $\hat{\mathcal{T}}^{l}_{0}\subseteq\mathcal{X}_{\eta}
	\subseteq\mathcal{X}$, the closed-loop trajectory remains inside the
	admissible state constraint set and the corresponding control inputs
	remain inside $\mathcal{U}$. Therefore, safety is preserved during the
	attack. Moreover, because the terminal RCI set is contained in
	$\mathcal{X}_{\eta}$, the networked tracking controller can be safely
	reactivated once the attack-free condition is restored, which guarantees
	recovery.
	
	Finally, for all $k\geq k'+N_l$, the state trajectory is confined in the
	terminal RCI region $\hat{\mathcal{T}}^{l}_{0}$. Therefore, the tracking
	error with respect to the reference signal satisfies
	\begin{equation}
			\|e_k\|=\|x_k-r_k\|
		\leq d^{sup}(\hat{\mathcal{T}}^{l_i}_0, r_k),
		\qquad \forall k\geq k'+N_l,
	\end{equation}
	which proves the stated tracking-error bound. This completes the proof.
	}
\end{proof}
%%%%%%%%%%%%%%% SECTION %%%%%%%%%%%%%%%%%%%%
\subsection{Tracking Supervisor Module (\textit{TS})} \label{sec:tracking_supervisor} 

If the anomaly detector module detects an anomaly, then the attack could be either on the actuation and/or measurement channel. In what follows, the \textit{TS} actions are derived assuming that the cyber-attack affects the measurement channel. In particular,  the objective of this module is to allow the networked control system to operate safely in a controlled open-loop mode (from the last attack-free state measurement) and minimize the tracking performance loss.

By taking into account the worst-case attack detection delay $\tau,$ if a cyber-attack is detected at $k'$, then the last reliable measurement is $x_{k'-\tau-1}.$ However, by exploiting forward reachability arguments, it is possible to robustly estimate from $x_{k'-\tau-1}$ the set of states $\hat{\mathcal{R}}_{k'}$  containing the current state $x_{k'},$ i.e., such that $x_{k'}\in\hat{\mathcal{R}}_{k'}.$ 
In a data-driven fashion, the set  $\hat{\mathcal{R}}_{k'}$ can be outer approximated by resorting to the following recursive RORS data-driven computation starting from the initial condition $\hat{\mathcal{R}}_{k'-\tau-1}= x_{k'-\tau-1},$
\begin{equation} \label{eq:forward_evolution_prediction}
	\begin{array}{rcr}
		\hat{\mathcal{R}}_{k'-\tau+t}\!\!\!\! &=&\!\!\!\!\! \mathcal{M}_{AB}
		\left [\mathcal{R}_{k'-\tau +t -1}, u_{k'-\tau +t -1}\right ] \oplus \mathcal{W},\\
		&& \forall\, t \in \Z_+.
	\end{array}	
\end{equation}

Given $\hat{\mathcal{R}}_{k'-\tau+t},$ \textit{TS} is instructed to replace $\forall k\geq k',$ ${x}_{k}$  with an admissible state $\hat{x}_{k} \in \hat{\mathcal{R}}_{k}.$ Consequently, under attack, the tracking controller will compute the following control action
\begin{equation}\label{eq:tracking_under_attack}
	u_{k} = \eta(\hat{x}_{k},r_k), \quad \forall \,k\geq  k'.
\end{equation}
{
	\begin{remark}\label{remark:computations}
		The computationally intensive operations required by Algorithm~\ref{algorithm:set_theoretic_MPC_and_SV}, including the construction of the Voronoi partition, the computation of the terminal RCI regions, and the recursive computation of the families of ROSC sets, are performed offline. The resulting sets are stored and used online only through region selection, set-membership checks, and the solution of ~\eqref{eq:control_based_extended_control_regions}. Therefore, the online complexity is significantly lower than the offline design complexity. The number of Voronoi regions $L$ defines a design trade-off: larger values of $L$ provide a finer partition of $\mathcal{X}_{\eta}$ and may reduce tracking performance degradation during attacks, but they also increase the offline computational burden and memory requirements because more RCI regions and ROSC families must be computed and stored. In practice, $L$ should be selected according to the desired tracking-performance granularity, and the available offline computational resources.
\end{remark}}
%%%%%%%%%%%%%%%%%%%%%%%%
\subsubsection{Tracking Performance Evaluation} \label{sec:tracking_performance_evaluation} 
To design the tracking supervisor logic to minimize the tracking performance loss under attack, we need to first offline approximately quantify the tracking performance degradation associated with the emergency controller actions. In particular, the tracking index $I(i, j)$ is here proposed:
\begin{equation} \label{eq:tracking_performance_index}
	I(i,j) = \alpha I_1(l_i, l_j) + \beta I_2(l_i, l_j), \quad l_i, l_j \in \mathcal{L},
\end{equation}
where $\alpha,\beta \geq 0$ are two weighting factors and 
\begin{itemize}
	\item $I_1(l_i, l_j) = d^{sup}(\hat{\mathcal{T}}^{l_i}_0, x^{l_j}_e)$ and $d^{sup}(\mathcal{S},p)$ computes the maximum distance between a point $p \in \rr^s$ and set $\mathcal{S} \subset \rr^s$  (see \cite[Definition 2]{gheitasi2022worst}). Such an index quantifies the nominal tracking error if $x \in \hat{\mathcal{T}}_0^{l_i}\subseteq \mathcal{V}_{l_i}$ and $r_{k'}$ belongs to $\mathcal{V}_{l_j.}$ In particular, $\hat{\mathcal{T}}^{l_i}_0$ is the RCI set where the state of the system will be confined in $N_{l_i}$ steps if the emergency controller is activated at the current time, and $ x^{l_j}_e$ is the disturbance-free equilibrium point of the partition containing $r_k.$
	\item $ \displaystyle I_2(l_i, l_j) \!=\!\!\! \min_{0 \leq p \leq N_{l_j}} p: \hat{\mathcal{T}}^{l_i}_0 \subseteq \bigcup_{s=0}^p \{\hat{\mathcal{C}}^{l_j}_s \}$, with $\{\hat{\mathcal{C}}^{l_j}_s \}_{s=0}^{N_{l_j}}$ a set of $N_{l_j} \geq 0$ ROSC set built as prescribed by \eqref{eq:ROSC-set} with starting RCI set $\hat{\mathcal{T}}^{l_j}_0 = \mathcal{V}_j$ and terminal condition $\hat{\mathcal{T}}^{l_i}_0 \subseteq \bigcup_{s=0}^{N_{l_j}}\{\hat{\mathcal{C}}^{l_j}_s \}.$ Such index quantifies the worst-case number of steps required for $x_{k}\in \!\hat{\mathcal{T}}^{l_i}_0 \subseteq \mathcal{V}_{l_i}$ to enter the Voronoi partition $\mathcal{V}_{l_j}$ containing $r_{k'}.$ 
\end{itemize}
\begin{remark} \label{remark:indexes}
	In simpler terms, by assuming a constant reference
	signal during the attack phase, $I_1(l_i, l_j)$ approximates the steady-state tracking error committed activating the emergency controller
	during the attack, while $I_2(l_i, l_j)$ approximates
	the time required to recover the reference tracking problem when the attack is terminated. $\hfill\square$
\end{remark}
Then, it is possible to sort all the pairs $(l_i, l_r), \forall l_i \in \mathcal{L}$ in an ascending order according to the tracking index	$I(l_i, l_r)$, i.e., 
\begin{equation} \label{eq:sorting_indexes}
	\begin{array}{c}
		\mathcal{I}(l_r) = \begin{bmatrix}
			I(l_1, l_r), \cdots,  I(l_x, l_r),  \cdots,   I(l_L, l_r) 
		\end{bmatrix} \\
		l_j \!\in\! \mathcal{L}, \forall j, I(l_1, l_r)\! \leq \! \cdots \!\leq \!  I(l_x, l_r), \!\leq\!  \cdots,  \!\leq\! I(l_L, l_r). 
	\end{array}
\end{equation}
Therefore, if $I(l_x, l_r) = I(l_1, l_r)$ then, the lowest tracking performance loss is obtained by forcing $x'_{k},$ $\forall\,k\geq k',$  to remain in $\mathcal{V}_{l_x}$. On the other hand, if $I(l_x, l_r) \neq I(l_1, l_r)$ then better tracking performance is obtained if $x'_{k}, k\geq k'$ can be steered into a pair $I(l_j, l_r)$ such that $I(l_j, l_r) < I(l_x, l_r).$

Since we can only robustly predict the robust reachable sets, $\hat{\mathcal{R}}_{k}$ prevents us from deterministically evaluating the tracking index $I(l_i, l_j)$ at the next time instant using a single vector approach. Consequently, the following index $J$ is defined:
\begin{equation} \label{eq:performance_index_weighted}
	\begin{array}{c}
		J_{k+1} = \! \displaystyle 
		\sum_{l_j \in \mathcal{L}} \!\!\frac{vol\left(\hat{\mathcal{R}}_{k+1} \cap \mathcal{V}_{l_j}\right)}{vol\left(\hat{\mathcal{R}}_{k+1} \right)} I(l_j,l_{r_{k}}) ,\quad k\geq k',
	\end{array}
\end{equation}
where $vol(\cdot)$ computes the volume of a set  
and $J_{k+1}$ defines a weighted sum of the tracking index based on the volume overlap between the uncertain prediction set and the Voronoi regions. Then, the open-loop tracking controller is kept active until one of the following stopping conditions is verified:
\begin{enumerate}
	\item $\hat{\mathcal{R}}_{k+1} \notin \mathcal{X}_{\eta}.$ Such a condition implies that $u_k,$ if applied to the plant, could bring the state of the plant outside of the admissible controller regions
	\item $J_{k+1}>J_{k}.$ This condition implies that the robust tube containing the state trajectory is moving toward regions with higher tracking performance loss (according to the index $I$).
\end{enumerate}
When one of the two conditions above arises, the \textit{TS} is instructed to replace $u_k$ with an invalid $u_k \in \mathcal{U}$ to activate the emergency controller intentionally (see Section~\ref{sec:safety_verification}).
Given the above results, the logic of the tracking supervisor has been summarized in Algorithm.~\ref{algorithm:tracking_supervisor}. 
{
	\begin{remark}\label{remark:conservatism}
		Note that the proposed solution resorts to worst-case arguments. Consequently, conservatism is an inherent feature of the proposed data-driven set-theoretic framework. In particular, the use of outer approximations $\hat{\mathcal{S}}_{k}^{+} \supseteq \mathcal{S}_{k}^{+}$ and of the matrix zonotope $\mathcal{M}_{AB}$ are instrumental to account for all system realizations consistent with the collected data and disturbance bounds, thereby ensuring that unsafe evolutions are not missed. At the same time, this introduce conservatism that may lead, for example, to an early activation of the Emergency Controller (EC), which can affect tracking performance. However, the Tracking Supervisor (TS) is designed to mitigate this effect by allowing controlled open-loop tracking based on reachable-set predictions and by delaying the activation of the EC as long as safety is certified and the performance index $J$ improves. Hence, the proposed architecture does not eliminate conservatism, which would compromise safety, but mitigates its impact on performance. 
		%This effect is also reflected in $J$, since larger reachable-set over-approximations modify the volume overlap between the predicted reachable set and the Voronoi regions.
	\end{remark}
}
\begin{algorithm}[h!] 
	\noindent
	\textit{Online ($\forall\, k$):}
	\begin{algorithmic}[1]
		\If{$d_k==1$}{}
		\State 
		Estimate  $\hat{\mathcal{R}}_k$ using \eqref{eq:forward_evolution_prediction} 
		and compute $u_k$ as in  \eqref{eq:tracking_under_attack}. 
		\If{($\hat{\mathcal{R}}_{k+1} \!\not \subseteq \! \!\mathcal{X}_{\eta}$) or ($J_{k+1}\!>\!J_{k}$)}{} \label{step_check}
		\State Replace $u_k$ with any $u_k\notin \mathcal{U}.$
		\EndIf
		\EndIf
		\State 	$u_{k}$ is sent.
	\end{algorithmic}
	\caption{Data-Driven Tracking Supervisor (\textit{TS})}
	\label{algorithm:tracking_supervisor}
\end{algorithm}
\begin{proposition} \label{proposition:performance_loss}
	Under any admissible cyber-attack, the data-driven tracking supervisor (Algorithm~\ref{algorithm:tracking_supervisor}) allows obtaining a tracking performance index ${I}$ better or equal to the one obtainable using only the safety verification module and emergency controller (Algorithm~\ref{algorithm:set_theoretic_MPC_and_SV}).
\end{proposition}
\begin{proof}
	{
	In the worst-case scenario, a cyber-attack may either remain undetected by
	the anomaly detector module~\eqref{eq:anomaly_detector} or corrupt the actuation channel. In these cases, the TS is either not activated or its computed action is invalidated by the
	attack. Hence, the closed-loop behavior is determined by the SV and EC
	modules, and the proposed architecture reduces to the baseline
	safety-preserving architecture without the TS.
	
	Consider now the case in which the attack affects the measurement channel
	and is detected at time $k'$. The TS then uses the last reliable
	measurement and the data-driven reachable-set recursion to compute
	$\hat{\mathcal{R}}_k$, and allows the tracking controller to operate in a
	supervised open-loop mode. This operation is allowed only while
	$\hat{\mathcal{R}}_{k+1}\subseteq\mathcal{X}_{\eta}$ and
	$J_{k+1}\leq J_k$. Therefore, during all time instants in which the TS
	keeps the tracking controller active, the predicted reachable tube remains
	inside the admissible tracking region and the tracking-performance index
	does not increase.
	
	Let $k_{\mathrm{stop}}$ be the first time instant at which one of the TS
	stopping conditions is violated. If $k_{\mathrm{stop}}=0$, the TS
	immediately triggers the EC, and the proposed scheme coincides with the
	baseline SV--EC architecture. If $k_{\mathrm{stop}}>0$, then the TS has
	kept the tracking controller active for $k_{\mathrm{stop}}$ time steps
	while preserving safety and satisfying $J_{k+1}\leq J_k$. When a stopping
	condition is violated, the TS intentionally sends an invalid input
	$u_k\notin\mathcal{U}$, which activates the EC through the SV module.
	Thus, after this time, the behavior again coincides with the baseline
	SV--EC architecture.
	
	Therefore, the TS either recovers the same behavior as the architecture
	without TS, or it safely delays the activation of the EC while the
	tracking-performance index does not deteriorate. Consequently, the
	tracking performance index obtained with the TS is better than or equal to
	the one obtained using only the SV and EC modules.
	}
\end{proof}
\begin{remark} \label{remark:using_MAC}
	The proposed solution has been developed assuming a detection delay $\tau>0.$ However,  the proposed solution preserves safety even if the actual detection delay for some attacks is larger than $\tau$. In such a scenario, the state predictions performed by the tracking supervisor might be incorrect (because starting from corrupted data) and lead to incorrect controller actions. Nevertheless, the presence of the SV module on the plant's side ensures that the plant's safety is always preserved.
	In addition, in some application scenarios, the communication channels between the plant and the controller
	are authenticated, i.e., a Message Authentication Code (MAC) \cite{menezes2018handbook} is
	used to authenticate every data packet sent over the network. Such a security mechanism allows the controller (or plant) to verify the authenticity and integrity of the received sensor measurements (or control actions), hence allowing instantaneous detection of network attacks. In this setup, the proposed solution can be straightforwardly adapted by simply setting $\tau=0.$
\end{remark}
%
%%%%%%%%%%%%%%%%%% SECTION %%%%%%%%%%%%%%%%%%%%
\section{Simulation} \label{sec:simulation}
In this section, we consider the industrial Continuous-Stirred Tank Reactor (CSTR) system used in \cite{guan2017distributed} as a testbed to evaluate the proposed approach. {Although the proposed framework is formulated for general finite-dimensional linear systems, the low-dimensional CSTR benchmark is intentionally selected to clearly illustrate the operation of the proposed set-theoretic mechanisms, including reachable sets, ROSC sets, and Voronoi partitions, without resorting to projections.} In this plant, chemical species $S_A$ reacts to form species $S_B,$ the state vector and control inputs are $x_p = [C_A, T]^T$ and $u=[T_C, C_{Ai}]^T$, where $C_A$ is the concentration of $S_A$ in the tank, $T$ the reaction temperature, $T_C$ the cooling medium temperature, and $C_{Ai}$ is input concentration of $S_A.$ A linearized discrete-time representation of the CSTR system using a sampling time $T_s=1\,\sec$ has been provided in \cite{guan2017distributed}, and it is characterized by the system matrices
\begin{equation} \label{eq:simulation_model}
	A = \begin{bmatrix}
		0.9719 & 0.0013 \\
		0.0340 & 0.8628
	\end{bmatrix}, B = \begin{bmatrix}
		-0.0839 & 0.0232 \\
		0.0761 & 0.4144
	\end{bmatrix},
\end{equation}
which we assume to be apriori unknown according to the considered data-driven setup. By considering a bounded disturbance set $\mathcal{W} = \{w: [-0.001, -0.001]^T \leq w \leq [0.001, 0.001]^T\}$ and assuming that the state and input constraints are $-2 \leq T_C \leq 2, -10 \leq C_{Ai} \leq 10, -10\leq C_A \leq 10,  -30 \leq T \leq 30,$ we have simulated the CSTR system for random input perturbation.
We have collected four input-state trajectories of five data samples each, which fulfill the rank condition  \eqref{eq:rank_condition}. Such data have been used to design, according to the scheme in  \cite{de2019formulas}, the networked  data-driven controller \eqref{eq:tracking_controller_logic}.
%
% we have collected 4 input-state trajectories that verify the rank condition \eqref{eq:rankCondition}. 
The estimated data-driven model $\mathcal{M}_{AB}$ and controller's parameters are available on the provided GitHub repository.

The emergency controller is configured with a five region Voronoi partition of $\mathcal{X}_{\eta}$ obtained using as generators the equilibrium states $x^1_e = [4,15]^T, x^2_e = [-6,15]^T, x^3_e = [0,0]^T, x^4_e = [6, -20]^T, x^5_e = [-4, -20]^T.$ 

On the other hand, the tracking supervisor is configured to use $\alpha = 1, \beta = 0$ in \eqref{eq:tracking_performance_index}.   Such choice instructs the tracking supervisor to evaluate the tracking performance degradation resulting from activating the emergency controller (the index $I_1(l_i, l_j)$) and not consider the time required to recover tracking when the attack is terminated (the index $I_2(l_i, l_j).$). Moreover, we have used the proposed data-driven anomaly detector introduced in \cite{attar2024safety},  and by simulating the system under different false-data injection attacks, a worst-case detection delay of $\tau=5$ has been obtained.

% {\color{red}It is assumed that the detection delay can be estimated using historical data, where attack detection algorithms are evaluated on past attack scenarios to measure the maximum time required to detect an attack. Consequently, by simulating the system under different false-data injection attacks, a worst-case detection delay of $\tau=5$ has been obtained.}
%
In the performed simulations, $x_0 = [0.01, -0.01]^T $ and the plant is required to track the time-varying reference signal $r_k$ shown in Fig~\ref{fig:trajectories} (see the blue stars), while three different cyber-attacks on the measurement channel occur. 
The first attack, for $60 \leq k \leq 125,$ injects $x^a_k=0.01[k-59,\, k-59]^T$ on the measurement vector $x_k$. The anomaly detector identifies the presence of an attack with a delay at $k=62$ and activates the tracking supervisor ($d_k=1$). Right after the attack has been detected, the tracking supervisor is instructed to assume  $x'_{56}$  as the last valid measurement to compute the predicted robust forward reachable sets (see red regions in Fig.~\ref{fig:trajectories}).
%
%\todo[inline]{what is $x_0$ in the simulation? show also it in the Figure}

The controlled open-loop evolution proceeds until $k=106$, when $\hat{\mathcal{R}}_{107} \not \subseteq \mathcal{X}_{\eta}$ and, the tracking supervisor, for safety reasons, activates the emergency controller (see Step~\ref{step_check} of Algorithm~\ref{algorithm:tracking_supervisor}). As a consequence for $107\leq k <129$, the emergency controller steers the state of the system into the RCI region $\mathcal{T}^1_0$
centered in $x^1_e$, and the reference tracking task is temporarily paused. 
However, when the first attack ends (at $k=110$), the plant recovers its tracking task (see Figs.~\ref{fig:trajectories}-\ref{fig:state_evolutions}). 
In the second scenario, an attack intermittently affects the measurement channel for $200\leq k < 261.$ In this case, the attacker injects $x^a_k=0.08[k-199,\, k-199]^T, x^a_k=0.1[k-239,\, k-239]^T$ during $k\in [200,\, 220]$ and $k\in [240,\, 260]$, respectively. The first attack is detected at $k=202$, and the second attack is detected at $k=241$. However, due to the nature of the attack, the sporadically received measurements enable the tracking supervisor to reset the uncertainty set (yellow sets in Fig.~\ref{fig:trajectories}), thereby preventing the suspension of the tracking task.
In the third attack scenario, the attacker injects $x^a_k = 0.1[k-399, \, k-399]^T$ on the measurements for $400\leq k \leq 420.$ 
In this case, the anomaly detector identifies the presence of an attack at $k=401$ and activates the tracking supervisor accordingly.
Differently from the first attack scenario, in this case, the uncertain predicted forward one-step evolution sets estimated by the tracking supervisor never violate the safety constraints (green sets in Fig.~\ref{fig:trajectories}). Moreover, the index $J$ presents a decreasing behavior, denoting that the tracking performance is improving under the action of the open-loop tracking controller's actions. 
%(see green sets in Fig.~\ref{fig:trajectories}).
As a consequence, the tracking task is never suspended under the actions of the cyber-attack.  

In Fig.~\ref{fig:state_evolutions} and Table.~\ref{table:comparison}, the proposed solution is contrasted with the one in \cite{attar2024safety} {and the model-based approach, using the nominal system matrices $A$ and $B,$  in \cite{gheitasi2022worst}}. Since in \cite{attar2024safety}, the emergency controller is activated regardless of the nature of the attack, an unavoidable tracking loss occurs in all the three considered attacks with a consequence of larger tracking performance loss. By measuring the tracking error, namely $e_r,$ as $e_r = \sum_{k=1}^{N_s} \frac{||x_k - r_k||}{N_s},$ with $N_s$ the simulation steps, Table~\ref{table:comparison} reports the obtained numerical results. It is possible to appreciate how, compared to \cite{attar2024safety}, the proposed solution reduces the tracking performance degradation due to the presence of cyber-attacks. {Moreover, the model-based approach in  \cite{gheitasi2022worst}  provides a lower tracking error, as expected, since it does not account for data-driven model uncertainty; nevertheless, the proposed data-driven architecture achieves a comparable tracking performance, with the remaining gap mainly attributable to the conservatism of the matrix-zonotope-based reachable-set approximation.}

{
To quantify the conservatism introduced by the data-driven reachable-set approximation, we compared, for each time instant, the volume of the used data-driven one-step robust forward reachable sets with its model-based counterpart  as in \cite{gheitasi2022worst}. 
%At each time instant $k$, the model-based (namely MB) and data-driven (namely DD) robust one-step reachable sets are given by
%$$
%\mathcal{S}^{+}_{k,\mathrm{MB}} = A x_k \oplus B u_k \oplus \mathcal{W},
%$$
%and
%$$
%\hat{\mathcal{S}}^{+}_{k,\mathrm{DD}} =
%\mathcal{M}_{AB}
%\begin{bmatrix}
%	x_k\\
%	u_k
%\end{bmatrix}
%\oplus \mathcal{W},
%$$
%respectively. 
%Their average volumes over the simulation horizon $T_{\mathrm{sim}}$ is computed as
%$$
%\bar{V}_{\mathrm{MB}}
%=
%\frac{1}{T_{\mathrm{sim}}}
%\sum_{k=1}^{T_{\mathrm{sim}}}
%\operatorname{vol}\left(\mathcal{S}^{+}_{k,\mathrm{MB}}\right),
%$$
%$$
%\bar{V}_{\mathrm{DD}}
%=
%\frac{1}{T_{\mathrm{sim}}}
%\sum_{k=1}^{T_{\mathrm{sim}}}
%\operatorname{vol}\left(\hat{\mathcal{S}}^{+}_{k,\mathrm{DD}}\right).
%$$
%The average volume gap is then reported as
%$$
%\Delta_{\mathrm{avg}}
%=
%100\left(
%\bar{V}_{\mathrm{DD}}-\bar{V}_{\mathrm{MB}}
%\right).
%$$
%In the simulation, 
The average model-based reachable-set volume is
$4.0\times 10^{-6}$, while the average data-driven reachable-set volume is
$7.0483\times 10^{-3}$. This corresponds to an average volume gap of $\Delta_{\mathrm{avg}}=0.7044\%$. This confirms that the data-driven reachable sets are larger, as expected, since they account for both bounded disturbances and the uncertainty in $A$ and $B$ through the matrix zonotope $\mathcal{M}_{AB}$.

To evaluate whether this conservatism leads to unnecessary controller switching, we also counted false-positive activations of the Emergency Controller.
% taking into account any case 
In particular, a false-positive activation is defined 
as %$\mathcal{S}^{+}_{k,\mathrm{MB}}\subseteq \mathcal{X}_{\eta}
%\quad \text{but} \quad
%\hat{\mathcal{S}}^{+}_{k,\mathrm{DD}}\nsubseteq \mathcal{X}_{\eta}$ (i.e., a case
an instance  
where the data-driven TS condition triggers the emergency mechanism only due to conservatism in the computed one-step forward reachable set.
%, while the corresponding model-based TS prediction remains admissible). 
The obtained results show the presence of only 1 false activation of the EC, corresponding to $0.92\%$ of the evaluated TS decision instants. 
% that only 1 false activation 
%Only one such TS-triggered false-positive activation was observed, 
These results show that, although the data-driven reachable sets are more conservative, the practical impact of this conservatism on unnecessary EC activation remains limited in the considered simulation.
%A false-positive safety-verification activation was defined as a time instant for which
}

{
Finally, we have investigated the effect of anomaly-detection delay $\tau$ on the tracking performance of the proposed data-driven architecture.
% In the performed simulations, the delay is taken into account  by initializing the TS from the last reliable state-input pair $(x_{k-\tau},u_{k-\tau})$ after an alarm is triggered. 
In particular, we have evaluated the tracking error of the proposed architecture under three possible worst-case detection delays, namely $\tau=0$ (no delay), $\tau=5$ and $\tau=7.$ The obtained results,  reported in
Table~\ref{table:delay_sensitivity}, show that the tracking error remains close for the three cases. Consequently, the tracking error does not increase monotonically with the detection delay and  this is expected because the closed-loop system involves switching between the tracking controller, the TS, and the EC. A smaller detection delay leads to earlier TS/EC intervention, which improves responsiveness to attacks but may also introduce earlier switching and slightly larger tracking deviations. 
%
%Therefore, the small non-monotonic variations observed in Table~\ref{table:delay_sensitivity} are mainly due to the interaction between attack timing, disturbance realizations, saturation, and switching decisions.
%
Overall, the results suggest that the proposed data-driven architecture preserves tracking performance with only limited degradation under moderate detection delays.
}
\begin{table}[htbp] 
	\centering
	
	\caption{Tracking Error: Proposed Approach, Model-Based Approach,  \cite{attar2024safety}, No Attack}
	\begin{tabular}[t]{lcccc}
		\hline
		& No attack & Proposed Approach & { \cite{gheitasi2022worst}} & \cite{attar2024safety}\\
		\hline
		$e_r$ & 1.22 & 1.57 & {1.24} & 6.66\\
		\hline
	\end{tabular}
	\label{table:comparison}
\end{table}%
\begin{figure}[htbp]
	\centering
	\includegraphics[width=\columnwidth]{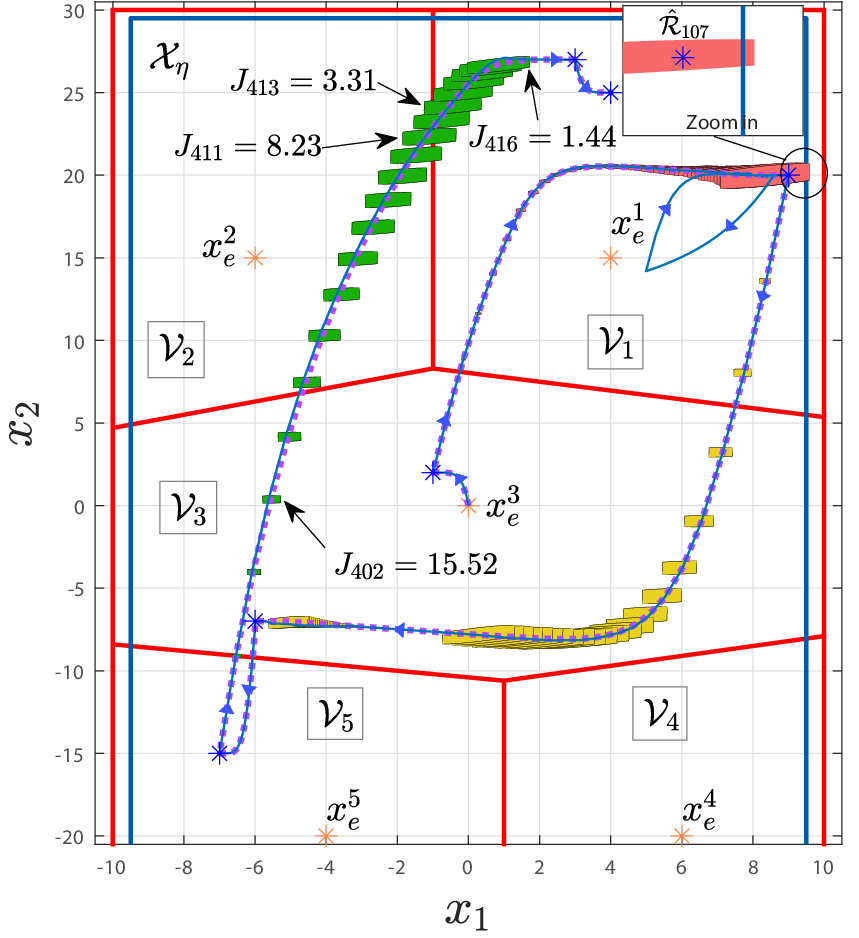}
	\caption{State trajectory: proposed solution with attacks (blue solid
		line) vs trajectory in attack-free scenario (purple dashed line).}
	\label{fig:trajectories}
\end{figure}%
\begin{table}[t]
	\centering
	\caption{Effect of detection delay on the tracking performance of the proposed data-driven architecture.}
	\label{table:delay_sensitivity}
	\begin{tabular}{c c}
		\hline
		Detection delay $\tau$ & Tracking error \\
		\hline
		0 & 1.582064 \\
		5 & 1.576347 \\
		7 & 1.578430 \\
		\hline
	\end{tabular}
\end{table}
\begin{figure}[htbp]
	\centering	\includegraphics[width=\columnwidth]{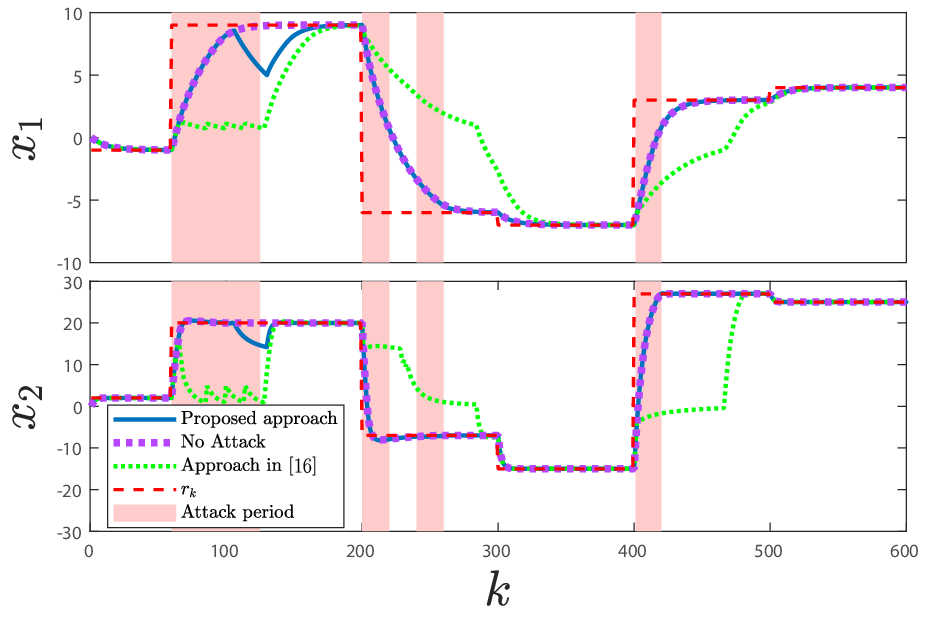}
	\caption{State evolution: no attack, proposed approach, \cite{attar2024safety}.}
	\label{fig:state_evolutions}
\end{figure}%
%
%%%%%%%%%%%%%%%%%%% SECTION %%%%%%%%%%%%%%%%%%%%
\section{Conclusions} \label{sec:conclusions}

This paper proposed a data-driven robust solution to the safety and reference tracking control problems for constrained CPSs by leveraging robust reachability arguments. The proposed control architecture included two data-driven add-on modules (local to the plant and to the networked controller) designed to ensure safety while enhancing, whenever possible, the tracking performance under cyber-attacks in a supervised manner. Theoretical and simulation results have been reported to show the effectiveness of the proposed scheme.  Future studies will be devoted to extending the proposed scheme to particular classes of nonlinear systems.

%
%\bibliographystyle{IEEEtran}
%\bibliography{bibliography}
%\bibliography{cas-refs}
\bibliographystyle{unsrt}
\bibliography{bibliography}
%
%\newpage
%\bio{images/Mehran Attar_photo.jpg}
%Mehran Attar received the Ph.D. degree in Information and Systems Engineering from Concordia University, Montréal, Canada, in 2024. He is currently a Researcher at Ericsson, Montréal. Before joining Ericsson, he was a postdoctoral researcher at ÉTS Montréal.
%%
%His research focuses on the development of artificial intelligence and data-driven methodologies for the safety, security, and resilience of cyber-physical systems, with applications in control, optimization, and cybersecurity. His broader research interests include machine learning, physical AI, quantum-inspired optimization, and intelligent decision-making systems.
%\bio{images/Lucia-photo.jpg}
%Walter Lucia is currently an Associate Professor at the Concordia Institute for Information Systems Engineering, Concordia University, Canada. He received the M.Sc. degree in Automation Engineering (2011) and the Ph.D. degree in Systems and Computer Engineering (2015) from the University of Calabria, Italy. Before joining Concordia University in 2016, he was a visiting research scholar in the ECE Department at Northeastern University (USA) and a visiting postdoctoral researcher in the ECE Department at Carnegie Mellon University (USA). Dr. Lucia is currently an Associate Editor for IEEE Control Systems Letters, Control System Society - Conference Editorial Board and IEEE Systems Journal.  Dr. Lucia's research interests include control of unmanned vehicles, model predictive control, and secure and resilient control of cyber-physical systems.

\end{document}